\documentclass[11pt]{article}

\usepackage{natbib}
\usepackage{amssymb,mathtools,bbm}
\usepackage{amsthm}
\usepackage{algorithm}
\usepackage{algpseudocode}
\usepackage{bm}
\usepackage{xcolor}
\usepackage{geometry}
\usepackage{setspace}

\newtheorem{theorem}{Theorem}
\newtheorem{lemma}{Lemma}
\newtheorem{assumption}{Assumption}

\newtheorem{remark}{Remark}

\newcommand{\E}{\mathbb{E}}

\newcommand{\1}{\mathbbm{1}}

\newcommand{\dd}{\mathrm{d}}

\title{Event Studies with Feedback\thanks{%
Botosaru: Department of Economics, McMaster University, email: \texttt{botosari@mcmaster.ca}. Liu: Department of Economics, University of Pittsburgh, email: \texttt{laura.liu@pitt.edu}. We thank Jeffrey Wooldridge and participants at the 2026 AEA annual meeting for helpful discussions. Botosaru gratefully acknowledges financial support from the Canada Research Chairs Program. All remaining errors are our own.}}
\author{Irene Botosaru \\
\textit{McMaster University} \and Laura Liu \\
\textit{University of Pittsburgh}}
\date{\today}

\begin{document}
\onehalfspacing
\maketitle

\begin{abstract}
  Event studies often conflate direct treatment effects with indirect effects operating through endogenous covariate adjustment. We develop a dynamic panel event study framework that separates these effects. The framework allows for persistent outcomes and treatment effects and for covariates that respond to past outcomes and treatment exposure. Under sequential exogeneity and homogeneous feedback, we establish point identification of common parameters governing outcome and treatment effect dynamics, the distribution of heterogeneous treatment effects, and the covariate feedback process. We propose an algorithm for dynamic decomposition that enables researchers to assess the relative importance of each effect in driving treatment effect dynamics.

\medskip
\noindent\textbf{Keywords: }Event study, heterogeneous treatment effects, dynamic panel data, sequential exogeneity, feedback mechanisms

\noindent\textbf{JEL classification: }C23, C21
\end{abstract}

\section{Introduction}

Event studies in the panel data setting aim to quantify how treatment effects evolve over time and vary across units. When outcomes are persistent and covariates adjust endogenously to past outcomes and treatment exposure, observed dynamic responses may not correspond to a single causal mechanism. Instead, they may combine direct effects with indirect effects operating through endogenous covariate adjustments. This distinction is empirically relevant when treatment triggers equilibrium adjustment. For example, a minimum wage policy may have a direct effect on wages but may also lead to changes in firm-level input demand, which may in turn affect individual wages. Separating these effects permits an assessment of how much of the dynamic response reflects direct versus indirect effects through equilibrium-induced covariate adjustment.

This paper develops a dynamic event study framework that separates these two effects. The \emph{direct structural effect} captures how outcomes respond to treatment holding the covariate path fixed, while the \emph{indirect adjustment effect} operates through covariates that evolve endogenously over time. We consider a panel event study setting with units $i = 1,\ldots,N$ observed over periods $t = 0,\ldots,T$. For exposition, we present a simplified version of the model. The outcome evolves in calendar time $t$ according to
\begin{equation}
\label{eq:outcome}
Y_{it} = \rho_Y Y_{i,t-1} + \alpha_i + X_{it}'\beta 
        + \sum_{j\in\mathcal{J}} D_{it}^j\delta_{ij} + U_{it},\quad U_{it} \stackrel{iid}{\sim} \mathcal N(0,\sigma^2_U),
\quad t = 1,\dots,T,
\end{equation}
where $Y_{it}\in\mathbb{R}$ is the scalar outcome, $X_{it}\in\mathbb{R}^K$ is a vector of time-varying covariates, and $\alpha_i$ captures unit-specific heterogeneity. Treatment is characterized by an event-time $t_{0i}\in\{1,\dots,T\}$. Let $\mathcal{J}\subset\mathbb Z$ denote a finite set of event-time indices for which dynamic treatment effects are specified. For each $j\in\mathcal{J}$, the event-time indicator $D_{it}^j = \1\{t - t_{0i} = j\}$ is a deterministic function of $t_{0i}$. The coefficients $\delta_{ij}$ represent unit-specific dynamic treatment effects at event time $j$ and capture the direct structural response to treatment.

Following \cite{botosaru2025time}, we impose a parsimonious dynamic structure on these heterogeneous treatment effects by assuming that they follow an autoregressive process in event time:
\begin{equation}
\label{eq:delta_ar}
\delta_{ij} = \rho_\delta \delta_{i,j-1} + \varepsilon_{ij},\quad \varepsilon_{ij} \stackrel{iid}{\sim} \mathcal N(0,\sigma^2_\varepsilon), \quad j=1,\dots,J,
\end{equation}
where $\rho_\delta$ is a common persistence parameter and $\varepsilon_{ij}$ are idiosyncratic innovations. Let $\lambda_i = (\alpha_i, \delta_{i0})'$ denote the vector of unobserved heterogeneity, and denote its conditional distribution by $$H(\lambda_i \mid \mathcal I_i^0), \quad \mathcal I_i^0 = \left\{Y_{i0},X_{i0},t_{0i}\right\}.$$ The framework readily extends to higher-order dynamics, continuous treatment intensities, and event-time leads for testing anticipation effects. Gaussianity in \eqref{eq:outcome} and \eqref{eq:delta_ar} is imposed for likelihood-based estimation, while the identification results do not require a correct Gaussian specification.

The parameter $\beta$ captures the indirect adjustment channel, mapping changes in
covariates into changes in outcomes. In many applications, these covariates are policy-reactive and adjust dynamically in response to treatment and past outcomes. As a result, treatment may affect outcomes not only directly, but also indirectly through sequentially exogenous covariate adjustments. Standard event study approaches typically rule out this feedback by imposing strict exogeneity of covariates. While \cite{botosaru2025time} relax this requirement to allow for sequential exogeneity, they treat the covariate process as given and do not model or identify the adjustment mechanism itself.

This paper extends their framework by explicitly modeling and identifying the covariate
feedback process. We allow covariates to be predetermined and impose a \emph{homogeneous
feedback} restriction, requiring the covariate adjustment rule to be common across units up to
observable conditioning variables. This restriction separates the two channels: the
distribution $H(\lambda_i \mid \mathcal I_i^0)$ governs unit-level heterogeneity in the direct
structural treatment effect, while $\beta$ and the identified feedback process capture the
indirect adjustment effects. The framework therefore permits the inclusion of policy-reactive
covariates and delivers an explicit decomposition of event study dynamics beyond
reduced-form comparisons, facilitating policy analysis in settings with sequentially exogenous covariate
adjustment.

\section{Model and Assumptions}

We denote the history of variables up to time $t$ as $Y_i^t = (Y_{i1}, \dots, Y_{it})'$ and $X_i^t = (X_{i1}', \dots, X_{it}')'$. Let $\theta = (\rho_Y, \rho_\delta, \beta, \sigma^2_U, \sigma^2_\varepsilon)'$ denote the vector of common structural parameters, and $\mathcal I_i^{t} = \left\{Y_i^{t},X_i^{t},\mathcal I_i^0\right\}$ denote the information set up to time $t$.

\begin{assumption}[Sequential Exogeneity]
\label{ass:seq_exog}
For all $t=1,\ldots,T$, conditional on $(\mathcal I_i^{t-1}, X_{it}, \lambda_i)$,
 $Y_{it}$ does not depend on future outcomes, covariates, or shocks; conditional on $(\mathcal I_i^{t-1},\lambda_i)$, $X_{it}$ does not depend on current shocks, nor on future outcomes, covariates, or shocks. Moreover, 
\[
\E[U_{it}\mid \mathcal I_i^{t-1}, X_{it}, \lambda_i]=0.
\]
\end{assumption}

To achieve point identification of the feedback process in short panels, we impose the following restriction on the feedback process, following \cite{bonhomme2025dynamics}.

\begin{assumption}[Homogeneous Feedback]
\label{ass:homog_feedback}
Let $f_t(\cdot\mid\cdot)$ denote the conditional density of $X_{it}$. For all
$t=1,\ldots,T$,
\begin{equation}
\label{eq:cov_adj}
f_t(X_{it} \mid \mathcal I_i^{t-1}, \lambda_i)
=
f_t(X_{it} \mid \mathcal I_i^{t-1})
\quad \text{a.s.},
\end{equation}
with the feedback factor defined below being positive on the support of $\mathcal I_i^T$:
\begin{equation}
\label{eq:g_def}
g(X_i^T \mid Y_i^T, \mathcal I_i^0)
=
\prod_{t=1}^T f_t(X_{it} \mid \mathcal I_i^{t-1}).
\end{equation}
\end{assumption}

Under Assumption \ref{ass:seq_exog}, covariates $X_{it}$ can depend on $\lambda_i$ through past outcomes in complex ways. Assumption~\ref{ass:homog_feedback} breaks this dependence by requiring the conditional covariate adjustment rule \eqref{eq:cov_adj} to be the same across individuals, regardless of their $\lambda_i$, conditional on the observable history $\mathcal I_i^{t-1}$. This restriction delivers a clean separation between the direct structural effect and the indirect adjustment effect. Conditional on $(X_i^T,\mathcal I_i^0)$, the likelihood kernel for $Y_i^T$ given $\lambda_i$ coincides with the representation studied in \cite{botosaru2025time}, while all information about the indirect adjustment effect is captured by $g(X_i^T \mid Y_i^T, \mathcal I_i^0)$. This factorization allows us to treat the feedback mechanism as a separately identified component and to apply the identification arguments of \cite{botosaru2025time} to the conditional distribution of $Y_i^T$ given $(X_i^T,\mathcal I_i^0)$ to identify the direct structural effect. The positivity condition on $g(X_i^T \mid Y_i^T, \mathcal I_i^0)$ ensures that this decomposition is well-defined.

\section{Identification}
\label{sec:identification}

We now establish that the model delivers point identification of both the direct and indirect effects. Specifically, we identify: (i) the common parameters $\theta$, (ii) the distribution of unobserved heterogeneity $H(\lambda_i \mid \mathcal I_i^0)$, which characterizes the direct structural effect, and (iii) the feedback process $f_t(X_{it} \mid \mathcal I_i^{t-1})$, which characterizes the indirect adjustment effect. 

\begin{theorem}[Identification]
\label{th:theorem}
Suppose $\left\{Y_{it},X_{it},\{D_{it}^j\}_{j\in\mathcal J}\right\}_{t=1}^T$ follow \eqref{eq:outcome} and
\eqref{eq:delta_ar}. Let Assumptions~\ref{ass:seq_exog}--\ref{ass:homog_feedback} hold, and suppose the regularity conditions of
\cite{botosaru2025time} for identification are satisfied (i.i.d.\
sampling over $i$, conditional independence of errors over calendar and event times,
nonvanishing and differentiable characteristic functions, and rank condition for the event study
design). Then, $\theta$, $H(\lambda_i\mid\mathcal I_i^0)$, and $\{f_t(X_{it}\mid \mathcal I_i^{t-1})\}_{t=1}^T$ are identified.
Moreover, the cohort-specific distribution $H(\lambda_i\mid t_{0i})$ and
unconditional distribution $H(\lambda_i)$ are identified.
\end{theorem}

\begin{remark}[Time Effects]
\label{rem:gamma_nuisance}
It is conventional to allow for additive time effects $\{\gamma_t\}$ in the outcome equation. Lemma \ref{lem:gamma_demeaning} in the Supplemental Appendix shows that such effects potentially entering \eqref{eq:outcome} can be removed by cross-sectional demeaning. Since this transformation leaves $\theta$ and $\lambda_i$ unchanged, Theorem \ref{th:theorem} continues to apply provided Assumptions \ref{ass:seq_exog} and \ref{ass:homog_feedback} are interpreted for the demeaned process $\left\{\dot Y_{it},\dot X_{it}\right\}$ and the corresponding initial conditions in $\dot{\mathcal I_i^0}$. Thus $\{\gamma_t\}$ are pure nuisance parameters that do not affect identification of $H(\lambda_i \mid \dot{\mathcal I_i^0})$.
\end{remark}

\section{Counterfactual Analysis and Dynamic Decomposition}

The main contribution of our framework is that it enables the decomposition of dynamic treatment effects into direct and indirect effects, which allows counterfactual analysis incorporating both channels. The identified objects $\theta$, $H(\lambda_i \mid \mathcal I_i^0)$, and $\{f_t(X_{it} \mid \mathcal I_i^{t-1})\}_{t=1}^T$ characterize both the heterogeneous treatment effects and the dynamic adjustment of covariates following treatment.

The factorization in the proof of Theorem \ref{th:theorem} implies that in the first step, the estimation of $\theta$ and $H(\lambda_i \mid \mathcal I_i^0)$ can proceed as in
\cite{botosaru2025time}. In a second step, the homogeneous feedback process for covariates can be estimated by modeling the transition densities $f_t(X_{it}\mid \mathcal I_i^{t-1})$, using, e.g.,
parametric Markov models or sieve methods. Under our assumptions, the parameters governing the direct channel $(\theta,H)$ and the feedback mechanism enter the likelihood in separable blocks, so estimation of $(\theta,H)$ need not rely on a particular parametric specification of the feedback model,
provided the feedback factor $g$ is estimated consistently on the relevant support. The feedback model can therefore be selected to balance flexibility and parsimony depending on the intended decomposition and counterfactual exercises.

Given an alternative treatment timing $t_{0i}^\ast$ and/or alternative initial conditions $\left\{Y_{i0}^\ast,X_{i0}^\ast\right\}$, one can construct joint counterfactual paths $\left\{Y_i^{T,\ast},X_i^{T,\ast}\right\}$ as described in Algorithm \ref{alg:counterfactual}. Starting from the counterfactual initial conditions $\mathcal I_i^{0,\ast}$, latent heterogeneity $\lambda_i$ is drawn from $H(\lambda_i\mid\mathcal I_i^{0,\ast})$ and counterfactual treatment effects are generated according to \eqref{eq:delta_ar}. The system is then iterated forward in calendar time, drawing covariates from the estimated feedback process and constructing outcomes using \eqref{eq:outcome}, which combines the direct effect through $\{\delta_{ij}^\ast\}$ and the indirect effect through $X_{it}^\ast$ and $\beta$.

The resulting counterfactual paths decompose dynamic event study responses into a direct structural component, driven by latent heterogeneity and treatment effect dynamics, and an indirect component operating through sequentially exogenous covariate adjustments. This decomposition is empirically relevant because it distinguishes between treatment effects that arise from heterogeneous structural responses versus those that operate through equilibrium adjustments. For example, in a minimum wage study, the direct effect captures how firm productivity responds to the wage change, while the indirect effect captures how firm adjustments in hours or employment composition feed back into productivity. The framework provides an approach to quantifying each effect, enabling researchers to assess their relative importance in driving event study dynamics and facilitating counterfactual analysis that propagates both effects.

\begin{algorithm}[t]
\caption{Simulation of Counterfactual Paths $\left\{Y_i^{T,\ast},X_i^{T,\ast}\right\}$}
\label{alg:counterfactual}
\begin{algorithmic}[1]
\State \textbf{Input:} $\theta$, $H(\lambda_i \mid \mathcal I_i^0)$, $\left\{f_t(X_{it} \mid \mathcal I_i^{t-1})\right\}_{t=1}^T$, alternative timing $t_{0i}^\ast$, and/or alternative initial conditions $\left\{Y_{i0}^\ast,X_{i0}^\ast\right\}$.
\State Set $\left\{Y_{i0}^\ast,X_{i0}^\ast\right\}$ to observed or chosen counterfactual initial conditions.
\State Draw $\lambda_i = (\alpha_i,\delta_{i0})'$ from $H(\lambda_i \mid \mathcal I_i^{0,\ast})$, with $\mathcal I_i^{0,\ast} = \left\{Y_{i0}^\ast,X_{i0}^\ast,t_{0i}^\ast\right\}$.
\For{$j=1$ to $J$}
\State Draw error term: $\varepsilon_{ij}^\ast\sim\mathcal{N}(0,\sigma^2_\varepsilon)$.
\State Update treatment effect:
\(
\delta_{ij}^\ast = \rho_\delta \delta_{i,j-1}^\ast + \varepsilon_{ij}^\ast.
\)
\EndFor
\For{$t=1$ to $T$}
    \State Draw covariate:
    \(
    X_{it}^\ast \sim f_t(\cdot \mid \mathcal I_i^{t-1,\ast}).
    \)
    \State Draw error term: $U_{it}^\ast\sim\mathcal{N}(0,\sigma^2_U)$.
    \State Update outcome:
    \(
    Y_{it}^\ast =
    \rho_Y Y_{i,t-1}^\ast
    + \alpha_i
    + (X_{it}^\ast)'\beta
    + \sum_{j\in \mathcal J} D_{it}^{j,\ast}\delta_{ij}^\ast
    + U_{it}^\ast.
    \)
\EndFor
\State \textbf{Output:} Counterfactual paths 
$\left\{Y_i^{T,\ast}, X_i^{T,\ast}\right\}$.
\end{algorithmic}
\end{algorithm}

\section{Extensions}

First, the homogeneous feedback framework is not tied to the linear model in \eqref{eq:outcome}. In principle, one can replace the outcome equation by a nonlinear panel model $Y_{it} \sim f_\theta(\cdot \mid \mathcal I_i^{t-1}, X_{it}, \lambda_i),$ where $\theta$ is a finite-dimensional parameter and $\lambda_i$ collects unobserved unit-specific heterogeneity. Examples include dynamic logit or probit models, Poisson or negative binomial count models, and Tobit models for censored outcomes. Under Assumptions \ref{ass:seq_exog}--\ref{ass:homog_feedback}, once we factor out the feedback term, the conditional distribution of $Y_i^T\mid X_i^T, \mathcal I_i^0$ depends on $(\theta, H)$ only through an average over $\lambda_i$. Identification of $(\theta,H)$ can then proceed by applying the relevant nonlinear identification results to this conditional distribution, while the feedback process $g(X_i^T \mid Y_i^T, \mathcal I_i^0)$ remains a separately identified component.
  
The homogeneous feedback assumption in \eqref{eq:cov_adj} can be relaxed to allow for observed group-specific covariate adjustment rules. That is, let $G_i$ denote an observed group indicator, such as industry, region, or demographic category, and assume that for each $t$, $f_t(X_{it} \mid \mathcal I_i^{t-1}, G_i, \lambda_i)
=
f_t(X_{it} \mid \mathcal I_i^{t-1}, G_i).$ That is, conditional on observable history \emph{and} group membership, the covariate adjustment process does not depend on unobserved heterogeneity, but may vary across
groups. Under this modification, the likelihood factorization applies group by group. In particular, for each $g$, $\theta$ and $H(\lambda_i \mid \mathcal I_i^0, G_i = g)$ are identified as before.

\bibliographystyle{apalike}
\bibliography{references_feedback}

@article{botosaru2025time,
  title={Time-Varying Heterogeneous Treatment Effects in Event Studies},
  author={Botosaru, Irene and Liu, Laura},
  journal={arXiv preprint},
  year={2025},
  note={Discussion paper}
}

@unpublished{bonhomme2025dynamics,
  title={Back to Feedback: Dynamics and Heterogeneity in Panel Data},
  author={Bonhomme, St\'ephane},
  note={Working Paper},
  year={2025}
}

\clearpage

\appendix
\section{Proof of Theorem \ref{th:theorem}}
\begin{proof}
The proof proceeds by factoring the joint likelihood of the data to separate the feedback process for $X_i^T$ from the structural outcome model for $Y_i^T$.

Let $\mathcal{L}(\lambda_i; Y_i^T, X_i^T \mid \mathcal I_i^0)$ denote the likelihood of the observed trajectory for unit $i$ conditional on $\mathcal I_i^0$. By the law of total probability, the conditional likelihood of the observables given $\mathcal I_i^0$ is
\begin{equation}
\label{eq:marginal_lik_I}
f(Y_i^T, X_i^T \mid \mathcal I_i^0)
= \int \mathcal{L}(\lambda_i; Y_i^T, X_i^T \mid \mathcal I_i^0) \dd H(\lambda_i \mid \mathcal I_i^0).
\end{equation}
Using Assumption \ref{ass:seq_exog}, we decompose $\mathcal L(\lambda_i; Y_i^T, X_i^T \mid \mathcal I_i^0)$ as a product over time. For $t\ge 1$ the transition kernels depend on the history through $(\mathcal I_i^{t-1}, X_{it}, \lambda_i)$, and conditioning additionally on $\mathcal I_i^0$ does not change the transition densities. Thus we can write
\begin{align*}
\mathcal{L}(\lambda_i; Y_i^T, X_i^T \mid \mathcal I_i^0)
&= \prod_{t=1}^T f(Y_{it} \mid \mathcal I_i^{t-1}, X_{it}, \lambda_i)
                   f_t(X_{it} \mid \mathcal I_i^{t-1}, \lambda_i) \\
&= \underbrace{\left( \prod_{t=1}^T f(Y_{it} \mid \mathcal I_i^{t-1}, X_{it}, \lambda_i) \right)}_{=\mathcal{L}^{SE}(\lambda_i; Y_i^T \mid X_i^T, \mathcal I_i^0)} \times
   \underbrace{\left( \prod_{t=1}^T f_t(X_{it} \mid \mathcal I_i^{t-1}, \lambda_i) \right)}_{=g(X_i^T \mid Y_i^T, \mathcal I_i^0)}.
\end{align*}
By Assumption \ref{ass:homog_feedback}, the feedback term $g(X_i^T \mid Y_i^T, \mathcal I_i^0)$ defined in \eqref{eq:g_def} does not depend on $\lambda_i$ and is identified from the joint distribution $f(Y_i^T, X_i^T \mid \mathcal I_i^0)$ via standard factorization: for each $t=1,\ldots,T$,
\[
f_t(X_{it} \mid \mathcal I_i^{t-1}) = \frac{f(Y_i^{t-1}, X_i^t \mid \mathcal I_i^0)}{f(Y_i^{t-1}, X_i^{t-1} \mid \mathcal I_i^0)},
\]
where the numerator and denominator are identified by marginalization, and the positivity condition in Assumption \ref{ass:homog_feedback} ensures the ratio is well-defined.

Substituting \eqref{eq:g_def} into \eqref{eq:marginal_lik_I} and rearranging, we define the reweighted conditional quasi-likelihood
\begin{equation}
\label{eq:reweighted_I}
\tilde{f}(Y_i^T \mid X_i^T, \mathcal I_i^0) = \frac{f(Y_i^T, X_i^T \mid \mathcal I_i^0)}{g(X_i^T \mid Y_i^T, \mathcal I_i^0)}
= \int \mathcal{L}^{SE}(\lambda_i; Y_i^T \mid X_i^T, \mathcal I_i^0) \dd H(\lambda_i \mid \mathcal I_i^0).
\end{equation}
The right hand side of \eqref{eq:reweighted_I} matches the representation in \cite{botosaru2025time}, so applying their identification argument under the stated regularity conditions establishes the identification of $\theta$ and $H(\lambda_i \mid \mathcal I_i^0)$. Finally, $H(\lambda_i \mid t_{0i})$ and $H(\lambda_i)$ are identified by integrating $H(\lambda_i \mid \mathcal I_i^0)$ over the distribution of $\mathcal I_i^0$ conditional on $t_{0i}$ and unconditionally, respectively.
\end{proof}

\section{Supplemental Appendix}
\setcounter{section}{0}

\begin{lemma}[Elimination of Time Effects]
\label{lem:gamma_demeaning}
Suppose the outcome equation \eqref{eq:outcome} is augmented with additive time effects,
\[
Y_{it} = \rho_Y Y_{i,t-1} + \alpha_i + X_{it}'\beta 
        + \sum_{j \in \mathcal J} D_{it}^j \delta_{ij} + \gamma_t + U_{it},
\quad t=1,\dots,T.
\]
Define the cross-sectional averages
\[
\bar{Y}_t = \frac{1}{N} \sum_{i=1}^N Y_{it}, \quad
\bar{X}_t = \frac{1}{N} \sum_{i=1}^N X_{it}, \quad
\bar{\alpha} = \frac{1}{N} \sum_{i=1}^N \alpha_i, \quad
\bar{U}_t = \frac{1}{N} \sum_{i=1}^N U_{it},
\]
and
\[
\bar{D}_t(\delta) = \frac{1}{N} \sum_{i=1}^N \sum_{j\in\mathcal J} D_{it}^j \delta_{ij}.
\]
Let the demeaned variables be
\[
\dot{Y}_{it} = Y_{it} - \bar{Y}_t, \quad
\dot{X}_{it} = X_{it} - \bar{X}_t, \quad
\dot{\alpha}_i = \alpha_i - \bar{\alpha}, \quad
\dot{U}_{it} = U_{it} - \bar{U}_t.
\]
Then the demeaned outcome satisfies
\begin{equation}
\label{eq:demeaned_outcome}
\dot{Y}_{it}
= \rho_Y \dot{Y}_{i,t-1} 
  + \dot{\alpha}_i
  + \left( \sum_{j \in \mathcal J} D_{it}^j \delta_{ij} - \bar{D}_t(\delta) \right)
  + \dot{X}_{it}'\beta
  + \dot{U}_{it},
\end{equation}
so that the time effects $\{\gamma_t\}$ are eliminated from the dynamic equation for $\dot Y_{it}$.
\end{lemma}

\begin{proof}
Averaging the augmented outcome equation over $i=1,\dots,N$ yields
\[
\bar{Y}_t 
= \rho_Y \bar{Y}_{t-1} 
  + \bar{\alpha} 
  + \bar{D}_t(\delta) 
  + \gamma_t 
  + \bar{X}_t'\beta 
  + \bar{U}_t.
\]
Subtracting this average from the individual equation gives
\begin{align*}
Y_{it} - \bar{Y}_t
&= \rho_Y (Y_{i,t-1} - \bar{Y}_{t-1}) + (\alpha_i - \bar{\alpha}) \\
&\quad + \left( \sum_{j \in \mathcal J} D_{it}^j \delta_{ij} - \bar{D}_t(\delta) \right)
     + (\gamma_t - \gamma_t)
     + (X_{it} - \bar{X}_t)'\beta
     + (U_{it} - \bar{U}_t).
\end{align*}
The term $(\gamma_t - \gamma_t)$ is identically zero. Using the definitions of the demeaned variables, this simplifies to
\[
\dot{Y}_{it}
= \rho_Y \dot{Y}_{i,t-1} 
  + \dot{\alpha}_i
  + \left( \sum_{j \in \mathcal J} D_{it}^j \delta_{ij} - \bar{D}_t(\delta) \right)
  + \dot{X}_{it}'\beta
  + \dot{U}_{it},
\]
which is exactly \eqref{eq:demeaned_outcome}. Hence the time effects $\{\gamma_t\}$ are eliminated from the dynamic equation for the demeaned outcome.
\end{proof}

\end{document}